\pdfoutput=1
\documentclass[11pt]{article}
\usepackage[preprint]{acl}
\usepackage{times}
\usepackage{latexsym}
\usepackage[T1]{fontenc}
\usepackage[utf8]{inputenc}
\usepackage{microtype}
\usepackage{inconsolata}
\usepackage{graphicx}
\usepackage{booktabs}
\usepackage{array}
\usepackage{multirow}
\usepackage{amsmath,amssymb,amsthm}
\usepackage{algorithm}
\usepackage{algpseudocode}

\usepackage{amsmath,amsfonts,bm}

\def\eqref#1{equation~\ref{#1}}

\def\1{\bm{1}}

\DeclareMathAlphabet{\mathsfit}{\encodingdefault}{\sfdefault}{m}{sl}
\SetMathAlphabet{\mathsfit}{bold}{\encodingdefault}{\sfdefault}{bx}{n}

\DeclareMathOperator*{\argmin}{arg\,min}

\newtheorem{theorem}{Theorem}
\newtheorem{corollary}{Corollary}
\newtheorem{lemma}{Lemma}[section]
\newtheorem{proposition}{Proposition}[section]
\theoremstyle{remark}
\newtheorem{remark}{Remark}[section]
\newcommand{\W}{W_1}
\hypersetup{
  pdftitle={ColNanoVDR: Document-Free Query Distillation for Multi-Vector Visual Document Retrieval via Optimal Transport},
  pdfauthor={Zhuchenyang Liu, Ziyi Wang, Yao Zhang, Yu Xiao}
}

\title{ColNanoVDR: Document-Free Query Distillation for Multi-Vector Visual Document Retrieval via Optimal Transport}

\author{
  Zhuchenyang Liu\textsuperscript{1} \quad
  Ziyi Wang\textsuperscript{2} \quad
  Yao Zhang\textsuperscript{1} \quad
  Yu Xiao\textsuperscript{1} \\
  \textsuperscript{1}Aalto University, Finland \\
  \textsuperscript{2}Independent Researcher, Netherlands \\
  \texttt{zhuchenyang.liu@aalto.fi} \\[6pt]
  {\normalfont\small
  Code:~\href{https://github.com/Ryenhails/NanoVDR}{\texttt{github.com/Ryenhails/NanoVDR}} \quad
  Models:~\href{https://huggingface.co/nanovdr}{\texttt{huggingface.co/nanovdr}}}
}

\begin{document}
\maketitle

\begin{abstract}
Multi-vector retrievers built on vision-language models lead visual document retrieval (VDR), but they run a multi-billion-parameter query encoder on every search. 
Distilling this encoder into a small student that queries the teacher's existing index would remove the bottleneck.
The standard recipe, however, matches the teacher's MaxSim scores and so requires encoding and caching every training page, which can reach terabytes of page tokens. 
NanoVDR avoids pages entirely by training on the teacher's query embeddings alone, but only for single-vector retrievers. We present ColNanoVDR, to our knowledge the first framework to bring this document-free distillation to multi-vector VDR. 
Its objective, OTW (Optimal Transport with Learned Weights), aligns the student's query tokens with the teacher's by entropic optimal transport, with a learned weight for each student token, and needs no correspondence between the two tokenizations.
We prove that the resulting alignment cost bounds the MaxSim score difference on every page. 
Distilled from five state-of-the-art teachers, the 149M text-only students retain about 95\% of their teachers' NDCG@5 on ViDoRe v1--v3 while encoding queries up to 26$\times$ faster. Under identical training, OTW matches score distillation while encoding no page and reading 12.6$\times$ less cached teacher data.
\end{abstract}

\section{Introduction}\label{sec:intro}
\begin{figure}[t]
  \centering
  \includegraphics[width=\linewidth]{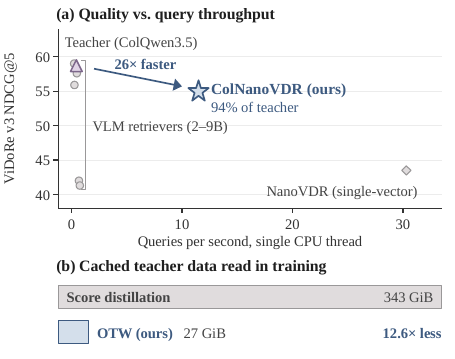}
  \caption{\textbf{ColNanoVDR approaches its teacher's quality at a fraction
  of its inference and training cost.} (a)~ViDoRe v3 NDCG@5 against query
  throughput on a single CPU thread. (b)~Cached teacher data read during
  training: score distillation reads query and page tokens, OTW only query
  tokens.}
  \label{fig:teaser}
\end{figure}

Visual document retrieval (VDR) matches textual queries directly against
page images, without relying on optical character recognition (OCR)
\citep{faysse2025colpaliefficientdocumentretrieval}. State-of-the-art
retrievers build both the query and document encoders on large
vision-language models (VLMs) with up to 9B parameters
\citep{loison2026vidorev3comprehensiveevaluation}. Pages are indexed
offline once, but the query encoder runs on every request, so a
multi-billion-parameter VLM sits on the online path of every search.
\textbf{Query-side distillation} removes this bottleneck: it replaces only
the teacher's query encoder with a small student and keeps the teacher's
document encoder and page index unchanged. NanoVDR
\citep{liu2026nanovdrdistilling2bvisionlanguage} showed that this works for
VDR with a text-only student trained on the teacher's query embeddings
alone, with no page encoded or read during training; we call such training
\textbf{document-free}.

However, NanoVDR targets single-vector retrievers, whereas the strongest
VDR systems are multi-vector: they represent queries and pages as sets of
token vectors and score them by late interaction, matching each query
token to its most similar page token and summing these maxima (MaxSim)
\citep{khattab2020colbertefficienteffectivepassage,faysse2025colpaliefficientdocumentretrieval}.
On ViDoRe v3, the multi-vector teachers we consider lead the best
single-vector system by 15 to 19 NDCG@5 points (Table~\ref{tab:main}).
Extending query-side distillation to multi-vector teachers would bring a
compact query encoder to the strongest retrievers and let it serve their
existing indices without re-indexing. We hypothesize that such a student
can retain most of its teacher's quality at a small fraction of its size
and latency.

The direct route is score distillation, the standard recipe for
late-interaction students
\citep{santhanam-etal-2022-colbertv2,huang-chen-2024-pairdistill,clavié2024jacolbertv25optimisingmultivectorretrievers},
which trains the student to reproduce the teacher's MaxSim scores on
candidate pages. It is document-dependent, however: every training page
must be encoded by the teacher, and each high-resolution page yields over
a thousand token vectors. For ColVec1.1~\citep{webai_colvec1_1_4b}, we
estimate that one million training pairs would require about 2\,TiB of
cached page tokens (Appendix~\ref{app:cache}), and the cache must be
rebuilt for every new teacher. Keeping training document-free, as in
NanoVDR, would remove this page-side overhead entirely. 

Compared with the single-vector case, document-free training for multi-vector retrievers raises two problems. 
First, the two encoders tokenize a query differently
and produce token sets of different sizes with no correspondence between them. Second, MaxSim takes a maximum
over page tokens, so it is not obvious that aligning query tokens controls
the score on pages never seen in training.

\begin{figure}[t]
  \centering
  \includegraphics[width=\linewidth]{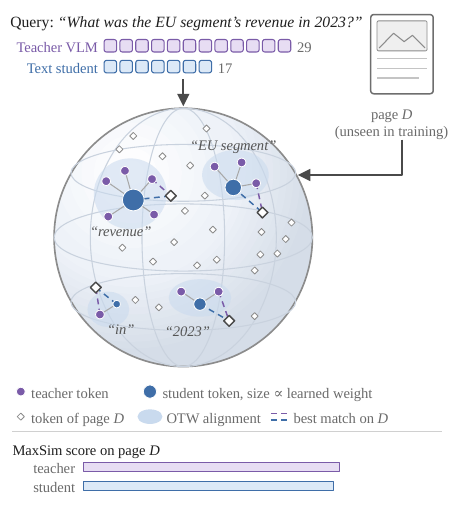}
  \caption{\textbf{Intuition behind OTW.} The teacher and the student
  encode the same query into token sets of different sizes on the unit
  sphere. OTW aligns the two sets softly (shaded regions): each student
  token covers nearby teacher tokens, and its learned weight grows with the
  number it covers. On a page $D$ never seen in training, aligned student
  and teacher tokens find the same best-matching page token, so the two
  MaxSim scores nearly coincide. Schematic.}
  \label{fig:otw-intuition}
\end{figure}

To address both problems, we propose \textbf{ColNanoVDR}, to our knowledge the first document-free, query-side distillation
framework for multi-vector VDR, trained with \textsc{OTW}
(\textbf{O}ptimal \textbf{T}ransport with Learned \textbf{W}eights). OTW
treats each query as a weighted set of token embeddings on the unit sphere
and aligns the student's set with the teacher's by entropic optimal
transport (Figure~\ref{fig:otw-intuition}). The transport plan couples
tokens without requiring a correspondence, and a lightweight linear head
predicts a weight for each student token from the query, so that one
student token can stand in for several teacher tokens; this addresses the
first problem. For the second, we show that the alignment cost between the
two token sets bounds the MaxSim score difference on every page
(Theorem~\ref{thm:w1}), so aligning queries suffices in principle and
training reads only the teacher's query tokens. At inference, the
student's tokens are rescaled by their weights and scored with standard
MaxSim against the teacher's unchanged index.

We distill five state-of-the-art multi-vector retrievers and evaluate on
ViDoRe v1, v2, and v3
\citep{faysse2025colpaliefficientdocumentretrieval,macé2025vidorebenchmarkv2raising,loison2026vidorev3comprehensiveevaluation}.
With 30$\times$--60$\times$ fewer parameters, ColNanoVDR retains about
95\% of its teacher's NDCG@5, and the ColQwen3.5 student encodes a query
26$\times$ faster than its teacher on a single CPU thread
(Figure~\ref{fig:teaser}a). Under identical training, OTW matches score
distillation while encoding no page and reading 12.6$\times$ less cached
teacher data (Figure~\ref{fig:teaser}b). Overall, ColNanoVDR offers a
practical path to deploying multi-vector visual document retrieval at
scale: queries are encoded on a single CPU thread and scored against the
teacher's existing index, and adding a new teacher requires only encoding
its training queries.

\section{Related Work}\label{sec:related}

\paragraph{Cost of multi-vector visual document retrieval.}
Late interaction scores a query against a document by matching every query token to its best document token \citep{khattab2020colbertefficienteffectivepassage}. ColPali brought it to page images \citep{faysse2025colpaliefficientdocumentretrieval}, and the ViDoRe benchmarks \citep{faysse2025colpaliefficientdocumentretrieval,macé2025vidorebenchmarkv2raising,loison2026vidorev3comprehensiveevaluation} have since driven multi-vector retrievers built on ever larger vision-language backbones \citep{tomoro2026colqwen3,athrael2026colqwen35}. Much of the work on the resulting overhead targets the index: storing fewer or merged document tokens \citep{hofstätter2022introducingneuralbagwholewords,Kankanampati_2026,ma2025storageefficientvisualdocumentretrieval,liu2026structuralanchorpruningtrainingfree}, or reducing multi-vector search to single-vector search \citep{dhulipala2026muveramultivectorretrievalfixed}. Another direction trains compact retrievers natively \citep{teiletche2025modernvbertsmallervisualdocument}, rebuilding both encoders. This requires re-indexing every collection, and the small document encoder limits quality. Neither direction addresses the overhead that remains once the multi-vector index is fixed. ColNanoVDR targets this overhead: it keeps the teacher's document encoder and index as they are and replaces only the query encoder, so pages keep the teacher's high-capacity representations, no collection is re-indexed, and the method is complementary to index-side compression.

\paragraph{Distillation for retrievers.}
Existing distillation recipes split along the axis that matters for our setting: whether training reads documents. Late-interaction students have been distilled through their scores on sampled query-document pairs, with a cross-encoder or pairwise reranker as the teacher \citep{santhanam-etal-2022-colbertv2,huang-chen-2024-pairdistill}, or by matching a strong teacher's ranking distribution over sampled documents \citep{clavié2024jacolbertv25optimisingmultivectorretrievers,takehi2025fantasticsmallretrieverstrain}; the same score-level signal has compressed late interaction into a single-vector student \citep{lin2020distillingdenserepresentationsranking}. All of these read documents during training, which in visual retrieval means pages that the vision-language teacher must first encode. Single-vector retrievers admit a document-free alternative: the student regresses the teacher's embedding directly. This has been used to distill full dual encoders \citep{yang2024clipkdempiricalstudyclip,lei-etal-2024-mcad} and, in the asymmetric setting where the two towers are parameterized separately \citep{dong2022exploringdualencoderarchitectures}, to replace only the query encoder against a frozen document encoder \citep{kim2023embeddistillgeometricknowledgedistillation,wang-hong-2023-query}, including NanoVDR, the closest prior work to ours, which does so with a text-only student for visual documents \citep{liu2026nanovdrdistilling2bvisionlanguage}. Regressing one vector onto another has no direct counterpart in late interaction, where each side is a set of token vectors of different sizes with no correspondence between them. ColNanoVDR supplies this counterpart, making query-side distillation document-free in the multi-vector setting.

\paragraph{Optimal transport and token weighting.}
Optimal transport has served as a distillation objective for aligning teacher and student distributions and representations: over labels \citep{bhardwaj-etal-2022-knot}, over the output distributions and hidden states of language models, including across different tokenizers \citep{10777837,cui2025multileveloptimaltransportuniversal,10.1609/aaai.v40i39.40619}, and over features within a batch \citep{chen2021wassersteincontrastiverepresentationdistillation}. Vision-language pretraining has aligned image patches with words at the token level, contrastively through a token-wise maximum similarity \citep{yao2021filipfinegrainedinteractivelanguageimage} or through a one-to-one matching used only during training \citep{nie2023lightcliplearningmultilevelinteraction}. Our use differs in what is transported: we transport between the two token-embedding sets from which the retrieval score itself is computed, which is what makes the alignment cost a bound on the score difference on every page (Theorem~\ref{thm:w1}). On the weighting side, late interaction sums query-token matches with equal weight, and a reproduction study traces the failure of multi-vector retrievers on long, narrative queries to this uniform weighting \citep{Ghosh_2026}. Proposed remedies attach weights to vocabulary items, set from corpus statistics or fitted on relevance labels \citep{s2025incorporatingtokenimportancemultivector}, or produce them with a gating module trained on relevance labels \citep{kang-etal-2025-trial}. Our weights are instead predicted per token from the query's context and learned without relevance labels or documents, as the student-side marginal of the transport plan; the same weights are kept at inference, so the quantity trained is the quantity served.

\begin{figure}[t]
\centering
\includegraphics[width=1\linewidth]{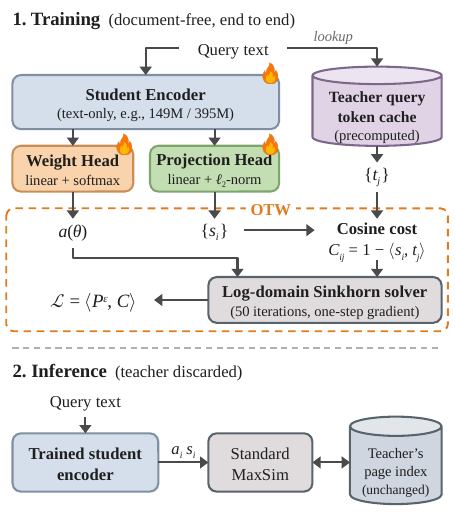}
\caption{The ColNanoVDR pipeline. \textbf{Training (top):} the student
outputs unit-norm query tokens $\{s_i\}$ and token weights $a(\theta)$;
OTW (dashed box) aligns them with the cached teacher query tokens
$\{t_j\}$ by entropic optimal transport, and the loss is the soft
alignment cost $\langle P^{\varepsilon},C\rangle$. 
\textbf{Inference (bottom):} the teacher is discarded; weighted
student tokens $a_i s_i$ are scored with standard MaxSim against the
unchanged teacher index.}
\label{fig:method}
\end{figure}

\section{Method}\label{sec:method}

ColNanoVDR keeps the teacher's document encoder and page index unchanged
and replaces only its query encoder (Figure~\ref{fig:method}). For a
query, the frozen teacher produces $K_t$ unit-norm token embeddings
$\{t_j\}_{j=1}^{K_t}$, computed once and cached; the student produces
$K_s$ unit-norm embeddings $\{s_i\}_{i=1}^{K_s}$ in the same space,
together with a weight for each token. The two models tokenize
differently, so $K_s\neq K_t$ in general and their tokens have no
correspondence. OTW trains the student by aligning the two token sets
with optimal transport. We first define the weighted score both models
share (Section~\ref{sec:notation}), then the alignment and its cost
(Section~\ref{sec:objective}), the learned weights and the training
objective (Section~\ref{sec:marginals}), and finally show why aligning
queries controls the score on every page (Section~\ref{sec:theory}).
Section~\ref{sec:inference} covers the architecture, solver, and
inference, and Appendix~\ref{app:proofs} proves every formal claim.

\subsection{Queries as weighted token sets}\label{sec:notation}

We represent a query by its unit-norm token embeddings $q_1,\dots,q_K$
with nonnegative weights $\omega_1,\dots,\omega_K$ summing to one, i.e., as
a weighted token set $\mu=\sum_i \omega_i\delta_{q_i}$, a discrete
probability measure in which $\delta_{q}$ is a unit point mass at $q$. For a page $D$, a set of unit-norm page
tokens, let $h_D(q)=\max_{d\in D}\langle q,d\rangle$ be the best match a
single query token $q$ finds on the page. The weighted late-interaction
score is the weighted average of these best matches,
\begin{equation}\label{eq:score-integral}
\bar S(\mu,D)\;=\;\sum_i \omega_i\,h_D(q_i).
\end{equation}
With uniform weights this is standard MaxSim divided by the query length,
which leaves the ranking unchanged. The two models differ in their
weights. The teacher keeps uniform weights $b=(b_1,\dots,b_{K_t})$ with
$b_j=1/K_t$, so $\mu_T=\sum_j b_j\delta_{t_j}$ is scored by its standard
MaxSim; the student uses learned weights $a=(a_1,\dots,a_{K_s})$, so
$\mu_S=\sum_i a_i\delta_{s_i}$ (Section~\ref{sec:marginals}). Distillation asks that
$\bar S(\mu_S,D)\approx\bar S(\mu_T,D)$ on every page $D$, without
seeing any page during training.

\subsection{Aligning token sets with optimal transport}\label{sec:objective}

\paragraph{Transport plans.}
We compare the two token sets through a soft alignment, much like a
word-alignment matrix in machine translation: a nonnegative matrix
$P\in\mathbb{R}^{K_s\times K_t}$ in which $P_{ij}$ is the share of student
token $i$'s weight assigned to teacher token $j$. Every student token
hands out exactly its weight $a_i$, and every teacher token receives
exactly its weight $b_j$:
\begin{equation}\label{eq:marginals}
\textstyle\sum_j P_{ij}=a_i,\qquad \sum_i P_{ij}=b_j .
\end{equation}
We call such a matrix a \emph{transport plan} and write $U(a,b)$ for the
set of them. A plan needs no correspondence between the two
tokenizations: one student token may cover several teacher tokens, and
one teacher token may be split across several student tokens.

\paragraph{Alignment cost.}
We measure the disagreement between two tokens by the cosine cost
$c(x,y)=1-\langle x,y\rangle$, built on the inner product that MaxSim
uses, and collect it in the cost matrix $C_{ij}=c(s_i,t_j)$. The cost of a plan, $\langle P,C\rangle=\sum_{ij}P_{ij}C_{ij}$, is
the average disagreement between aligned tokens, and the lowest cost over
all plans,
\begin{equation}\label{eq:alignment-cost}
\mathrm{OT}_c(\mu_S,\mu_T)\;=\;\min_{P\in U(a,b)}\;\langle P,C\rangle ,
\end{equation}
is an earth mover's problem between the two token sets, the formulation
behind Word Mover's Distance \citep{kusner2015word}, here posed between
two encoders' embeddings of the same query. We refer to it as the
\textbf{alignment cost}. It is small when every teacher token has student
weight close to it, and Section~\ref{sec:theory} shows that it bounds the
MaxSim score difference on every page.

\paragraph{Entropic smoothing.}
The optimal plan of Equation~\ref{eq:alignment-cost} solves a linear
program: it is sparse and can jump between alignments under small changes
of the embeddings, a poor target for gradient training. We add an entropy
term \citep{cuturi2013sinkhorn,peyré2020computationaloptimaltransport},
which spreads weight over plausible partners:
\begin{equation}\label{eq:entropic-ot}
P^{\varepsilon}(a)\;=\;\argmin_{P\in U(a,b)}\;\langle P,C\rangle-\varepsilon H(P),
\end{equation}
where $H(P)=-\sum_{ij}P_{ij}\log P_{ij}$. We write $P^{\varepsilon}(a)$
because the student weights $a$ are learned, whereas the teacher weights
$b$ are fixed. The strength $\varepsilon$ acts
as a temperature: as $\varepsilon\to 0$ the plan approaches the optimal
one, and a large $\varepsilon$ spreads each token's weight evenly. The
plan is unique, differentiable in $C$ and $a$, and computed with Sinkhorn
iterations (Section~\ref{sec:solver}).

\subsection{Learned token weights and the OTW objective}\label{sec:marginals}

The alignment cost depends on the student's weights. Uniform weights,
which recover plain MaxSim, fit poorly whenever the two encoders tokenize
a query differently. This is the general case for query-side
distillation: the student and teacher backbones differ in vocabulary,
segmentation, and special tokens such as ColBERT-style query augmentation.
On our training queries, for instance, the student produces 17 tokens on
average against the teacher's 29 (Appendix~\ref{app:tokens}). Some
student tokens must then stand in for several teacher tokens, while others
have little to align with, yet uniform weights make every student token
hand out the same mass, which keeps the alignment cost high. Rather than
engineering the two tokenizers into correspondence for each
teacher--student pair, we let the student learn its own weights, a
lightweight design that applies to any pair.

If the weights could be chosen freely to minimize the alignment cost,
each teacher token would send its mass to its nearest student token, and
a student token's weight would become the share of teacher tokens for
which it is nearest (Proposition~\ref{prop:semirelaxed}). However, these
weights depend on the teacher's tokens, which are not available at
inference. We therefore train the student to predict its own weights from
the query alone. A linear head $w_\theta$ reads each token's hidden state
$z_i$ before the projection, and a softmax over the query's tokens turns
the scores into weights,
\begin{equation}\label{eq:weights}
a(\theta)\;=\;\mathrm{softmax}\bigl(w_\theta(z_1),\dots,w_\theta(z_{K_s})\bigr),
\end{equation}
where $\theta$ collects all student parameters: the encoder, the
projection, and the weight head.

\paragraph{Training objective.}
The \textsc{OTW} loss is the soft alignment cost of the entropic plan
under the student's predicted weights,
\begin{equation}\label{eq:final-objective}
\mathcal{L}(\theta)\;=\;\bigl\langle P^{\varepsilon}\bigl(a(\theta)\bigr),\;C(\theta)\bigr\rangle
\;=\;\textstyle\sum_{ij}P^{\varepsilon}_{ij}\,C_{ij},
\end{equation}
minimized end to end over the encoder, the projection, and the weight
head. The weight head receives no direct supervision: through the loss,
each student token moves toward the teacher tokens aligned with it, and
weight shifts toward student tokens that lie close to many teacher tokens.
Every term of
$\mathcal{L}$ is computed from the two query token sets, so training
requires no pages and no document cache.

\begin{table*}[t]
\centering
\footnotesize
\begin{tabular}{ll r@{~}l r@{~}l r@{~}l}
\toprule
Model & Params & \multicolumn{2}{c}{v1} & \multicolumn{2}{c}{v2} & \multicolumn{2}{c}{v3} \\
\midrule
\multicolumn{8}{l}{\textit{Reference systems (native retrieval)}} \\
\quad Tomoro-ColQwen3-8B       & 8.8B & 90.6 & & 65.0 & & 59.0 & \\
\quad ColVec1.1-8b             & 8.4B & 91.5 & & 67.8 & & 62.6 & \\
\quad ColNomic-7B              & 7.8B & 89.8 & & 60.4 & & 55.9 & \\
\quad ColQwen3.5-4.5B          & 4.5B & 91.6 & & 63.7 & & 58.7 & \\
\quad Vultron-4.5B             & 4.5B & 91.8 & & 67.6 & & 61.0 & \\
\quad ColVec1.1-4b             & 4.5B & 90.7 & & 66.6 & & 61.6 & \\
\quad Tomoro-ColQwen3-4B       & 4.4B & 90.2 & & 65.3 & & 57.6 & \\
\quad DSE-Qwen2                & 2.2B & 85.1 & & 55.7 & & 41.3 & \\
\quad ColPali-v1.3             & 3.0B & 84.2 & & 54.7 & & 42.0 & \\
\quad ColModernVBert           & 259M & 76.7 & & 33.4 & & 17.4 & \\
\quad NanoVDR-S (single-vector) &  69M & 82.2 & & 60.5 & & 43.5 & \\
\midrule
\multicolumn{8}{l}{\textit{ColNanoVDR (149M text-only student, document-free OTW), by teacher}} \\
\quad from ColQwen3.5-4.5B      & 149M & 90.7 & (99.0) & 60.0 & (94.2) & 55.1 & (93.8) \\
\quad from Tomoro-ColQwen3-8B   & 149M & 90.0 & (99.3) & 60.6 & (93.3) & 54.9 & (93.0) \\
\quad from Vultron-4.5B         & 149M & 91.3 & (99.4) & 65.0 & (96.1) & 58.3 & (95.5) \\
\quad from ColVec1.1-4b         & 149M & 90.3 & (99.5) & 64.0 & (96.1) & 59.1 & (95.8) \\
\quad from ColVec1.1-8b         & 149M & 90.9 & (99.3) & 65.4 & (96.5) & 60.1 & (96.0) \\
\bottomrule
\end{tabular}
\caption{Main results. NDCG@5 per benchmark and, for ColNanoVDR, retention
of its own teacher in parentheses, each student scored against that
teacher's index. Reference systems are evaluated by us under the identical
protocol (model identifiers in Appendix~\ref{app:eval}).}
\label{tab:main}
\end{table*}

\subsection{Why aligning queries suffices}\label{sec:theory}

The OTW objective aligns the student's query tokens with the teacher's,
but it is not obvious that this also aligns their MaxSim scores: the score
takes a maximum over the tokens of a page, and training never sees a page.
Viewing each query as a discrete measure on the unit sphere
(Section~\ref{sec:notation}), we bound the score difference directly by
the quantities OTW computes (proof in Appendix~\ref{app:proof-thm1}).

\begin{theorem}\label{thm:w1}
For every non-empty finite page $D$ on the unit sphere,
\[
\begin{aligned}
&\bigl|\bar S(\mu_S,D)-\bar S(\mu_T,D)\bigr|\\
&\quad\le\sqrt{2\,\mathrm{OT}_c(\mu_S,\mu_T)}
\le\sqrt{2\,\mathcal{L}(\theta)}.
\end{aligned}
\]
\end{theorem}

In plain terms, the objective OTW minimizes is an upper bound on the
MaxSim score difference between student and teacher on every page,
including pages absent from training. OTW thus provides a direct
sufficient condition for retrieval fidelity.

\subsection{Architecture, solver, and inference}\label{sec:inference}
\label{sec:solver}

\paragraph{Architecture.}
The student is a text-only encoder with two linear heads on its token
states (Figure~\ref{fig:method}, top): a bias-free projection to the
teacher's width followed by $\ell_2$ normalization, which yields
$\{s_i\}$, and the weight head of Equation~\ref{eq:weights}
(Appendix~\ref{app:training}). The teacher's query tokens are cached once,
so the teacher is never run during training.

\paragraph{Solver.}
We solve Equation~\ref{eq:entropic-ot} with log-domain Sinkhorn
iterations. The first iterations run without gradient tracking; only the
last one is recomputed inside the autograd graph, and gradients flow
through it to both the embeddings (through $C$) and the weights (through
$a$)
\citep{luise2018differentialpropertiessinkhornapproximation,eisenberger2022unifiedframeworkimplicitsinkhorn}.
Backpropagation thus needs the memory of a single iteration. On
query-sized matrices (about $17\times 29$), training with the solver is
about as fast as score distillation (Appendix~\ref{app:solver},
Algorithm~\ref{alg:otw}).

\paragraph{Inference.}
At inference time the teacher is discarded (Figure~\ref{fig:method},
bottom). The student scales each token by its weight,
$\tilde s_i=a_i s_i$, and MaxSim then returns
\[
\textstyle\sum_i\max_{d\in D}\langle\tilde s_i,d\rangle
=\sum_i a_i\,h_D(s_i)=\bar S(\mu_S,D),
\]
since positive weights move out of the maximum
(Proposition~\ref{prop:inference}). This is exactly the weighted score that
training aligns with the teacher's, so the student plugs into an existing
late-interaction engine with no change to its index or scoring kernel.

\section{Experiments}\label{sec:experiments}

Our experiments answer four questions. \textbf{Fidelity:} how much of a
multi-vector teacher's retrieval quality does a document-free student
retain, across teachers and student sizes? \textbf{Efficiency:} what does
the student save in query latency at inference and in cached teacher data
during training? These two are answered in Section~\ref{sec:results}.
\textbf{Objective:} under identical training, does OTW match
document-dependent score distillation, and do its learned weights matter
(Section~\ref{sec:mechanism})? \textbf{Deployment:} does the student stay
compatible with a compressed index (Section~\ref{sec:compression})?

\subsection{Setup}\label{sec:setup}

\paragraph{Teachers.}
Because OTW needs the teacher only on the training queries, adding a teacher is cheap, which makes a multi-teacher study feasible. We distill the five strongest multi-vector retrievers on ViDoRe v3 \citep{loison2026vidorev3comprehensiveevaluation} under our protocol (upper block of Table~\ref{tab:main}), spanning two backbone generations, two embedding dimensions (320 and 640), and 4.5B to 8.8B parameters. ColQwen3.5-4.5B, our main teacher, is abbreviated ColQwen3.5 in the text.

\paragraph{Students.}
We use the Ettin encoder suite \citep{weller2026seqvsseqopen}, pretrained with one recipe across all sizes, so that capacity is the only scaling variable, with the two heads of Section~\ref{sec:inference}. Its tokenizer differs from every teacher's, the general case that OTW targets. Main results use Ettin-150M, which with its heads gives a 149M-parameter student; a capacity ablation covers Ettin-32M, -68M, -150M, and -400M.

\paragraph{Training data and optimization.}
We adopt the NanoVDR training set \citep{liu2026nanovdrdistilling2bvisionlanguage}\footnote{\url{https://huggingface.co/datasets/nanovdr/NanoVDR-Train}}: 711{,}603 (query, page-image) pairs from four public sets, plus 777{,}649 machine-translated query variants that reuse the base pairs' pages (Appendix~\ref{app:data}). OTW and every alternative objective of Section~\ref{sec:mechanism} are trained on identical data with identical optimization; hyperparameters and transport settings are given in Appendix~\ref{app:training}.

\paragraph{Evaluation.}
We measure retrieval quality by NDCG@5, the standard ViDoRe metric, averaged within each benchmark
\citep{faysse2025colpaliefficientdocumentretrieval,macé2025vidorebenchmarkv2raising,loison2026vidorev3comprehensiveevaluation}: v1 (10 datasets), v2
(4 datasets), and v3 (8 datasets), together with retention, the ratio of
the student's to the teacher's benchmark-average NDCG@5 under the same
index, computed before rounding. The v3 suite consists of enterprise
collections (finance, human resources, industrial, pharmaceutical,
physics, computer science, energy), largely outside the training domains, with
queries authored against long multi-page documents. Efficiency is measured
by single-query encoding latency on CPU and GPU and by the volume of cached
teacher data read during training.

\paragraph{Comparisons.}
We compare each student with its own teacher and with ten further
retrievers evaluated under the identical protocol
(Table~\ref{tab:main}): multi-vector vision-language retrievers from
ColPali-v1.3 to the current state of the art, the single-vector
DSE-Qwen2, and two compact retrievers, ColModernVBert and the
single-vector NanoVDR-S. To isolate the objective,
Section~\ref{sec:mechanism} retrains the same student under two
document-dependent and two document-free alternatives.

\subsection{Main results}\label{sec:results}

\begin{table}[t]
\centering
\footnotesize
\setlength{\tabcolsep}{3.5pt}
\begin{tabular}{@{}lrrrr@{}}
\toprule
 & & \multicolumn{2}{c}{Encode (ms)} & \\
\cmidrule(lr){3-4}
Query encoder & Params & CPU & GPU & v3 \\
\midrule
\multicolumn{5}{@{}l}{\textit{Vision-language retrievers}} \\
Tomoro-ColQwen3-8B  & 8.8B & 4{,}277 & 18.4 & 59.0 \\
ColNomic-7B         & 7.8B & 3{,}845 & 23.4 & 55.9 \\
ColQwen3.5-4.5B (teacher) & 4.5B & 2{,}290 & 110.7$^\dagger$ & 58.7 \\
Tomoro-ColQwen3-4B  & 4.4B & 2{,}118 & 18.6 & 57.6 \\
ColPali-v1.3        & 3.0B & 1{,}504 & 17.0 & 42.0 \\
DSE-Qwen2           & 2.2B & 1{,}343 & 13.9 & 41.3 \\
\midrule
\multicolumn{5}{@{}l}{\textit{Compact retrievers}} \\
ColModernVBert      & 259M &  89 & 13.8 & 17.4 \\
NanoVDR-S (single-vector) &  69M &  33 &  1.7 & 43.5 \\
ColNanoVDR-400M     & 395M & 286 & 13.2 & 56.7 \\
\textbf{ColNanoVDR} & 149M &  87 & 10.8 & 55.1 \\
\bottomrule
\end{tabular}
\caption{Query-encoding cost on one node, median over 20 queries at batch size~1, excluding MaxSim scoring: one CPU thread (float32) and one H200 (bf16); v3 is ViDoRe v3 NDCG@5. ColNanoVDR students of the same size share the encoder, so one row per size; v3 is that of the ColQwen3.5 student. $^\dagger$Torch fallback for the hybrid linear-attention layers. Protocol in Appendix~\ref{app:latency}.}
\label{tab:serving}
\end{table}

\paragraph{Multi-vector quality is retained.}
In Table~\ref{tab:main}, each row of the lower block is a separate 149M student scored on its own teacher's index. On v1, every student retains about 99\% of its teacher's NDCG@5; on the harder v2 suite and the out-of-domain v3 enterprise collections, retention stays between 93.0\% and 96.5\% for all five teachers. A text-only student with 30$\times$--60$\times$ fewer parameters, trained without any document-side supervision, thus preserves most of its vision-language teacher's ranking quality.

\paragraph{Inference latency.}
Table~\ref{tab:serving} and Figure~\ref{fig:teaser}a measure the query path on one node. On a single CPU thread, ColNanoVDR encodes a query in 87\,ms, 26$\times$ faster than its ColQwen3.5 teacher with 30$\times$ fewer parameters, while every vision-language retriever needs more than a second per query. On a GPU at batch size~1 the gap is smaller, 1.3--2.2$\times$ against the other vision-language retrievers; ColQwen3.5's 110.7\,ms reflects a fallback kernel path rather than its model size.

\paragraph{Supervision cost.}
OTW reads the teacher's query tokens only. Score distillation reads, in addition, the cached tokens of the training pages and in-batch negatives at every step, 12.6 times as much cached teacher data in total (measured in Table~\ref{tab:cache} in Appendix~\ref{app:cache}). The document-free recipe removes this page-side cost for every new teacher. Per epoch, OTW trains at a speed comparable to score distillation (Appendix~\ref{app:training}), so the saving lies in storage and reads rather than computation.

\paragraph{Capacity.}
Across student capacity (Figure~\ref{fig:efficiency_scale}, numbers in Appendix~\ref{app:capacity}), retention rises
monotonically from Ettin-32M to Ettin-400M on every benchmark, with the
gains concentrated on v2 and v3. The Ettin-32M student already
retains 87\% of its teacher on v3. On v3, the retention curves of the two teachers
agree to within 1.0 point at every size.

\begin{figure}[t]
\centering
\includegraphics[width=\linewidth]{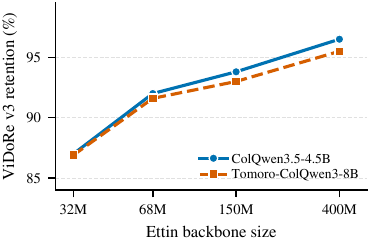}
\caption{ViDoRe v3 retention across the Ettin family, for the ColQwen3.5 and the Tomoro-ColQwen3-8B students.}
\label{fig:efficiency_scale}
\end{figure}

\begin{table*}[t]
\centering
\footnotesize
\setlength{\tabcolsep}{5pt}
\begin{tabular}{llccccccc}
\toprule
 & & & \multicolumn{3}{c}{ColQwen3.5-4.5B teacher} & \multicolumn{3}{c}{Tomoro-ColQwen3-8B teacher} \\
\cmidrule(lr){4-6}\cmidrule(lr){7-9}
Objective & Reads & Weights & v1 & v2 & v3 & v1 & v2 & v3 \\
\midrule
\multicolumn{9}{l}{\textit{document-dependent}} \\
\quad InfoNCE             & D   & uniform & 88.9 & 51.7 & 47.1 & 87.9 & 52.1 & 46.2 \\
\quad Listwise KL         & Q+D & uniform & \textbf{90.9} & \textbf{61.2} & 55.0 & 89.0 & 57.3 & 49.9 \\
\multicolumn{9}{l}{\textit{document-free}} \\
\quad Coverage            & Q   & uniform & 90.2 & 58.3 & 53.9 & 88.5 & 57.0 & 53.0 \\
\quad OT-uniform          & Q   & uniform & 89.6 & 58.7 & 53.4 & 84.4 & 51.4 & 50.2 \\
\quad \textsc{OTW} (ours) & Q   & learned & 90.7 & 60.0 & \textbf{55.1} & \textbf{90.0} & \textbf{60.6} & \textbf{54.9} \\
\bottomrule
\end{tabular}
\caption{Training objectives on two teachers (NDCG@5; same 149M Ettin
student and training setup throughout). ``Reads'': cached teacher data
per step, query tokens (Q), document tokens with in-batch negatives (D),
or both (Q+D). Bold: best per column.}
\label{tab:objectives}
\end{table*}

\subsection{Comparison of training objectives}\label{sec:mechanism}

We isolate the contribution of the objective by retraining the student under four alternative objectives with everything else fixed, on both ColQwen3.5 and Tomoro-ColQwen3-8B.

\paragraph{The objectives.}
Two baselines are document-dependent. \textbf{InfoNCE} contrasts the student's MaxSim score on the paired page with its scores on the other in-batch pages. \textbf{Listwise KL} matches the softmax over the student's in-batch MaxSim scores to the teacher's.

The other two are document-free and, like OTW, read only teacher query tokens, one query at a time. \textbf{Coverage} minimizes the average cosine distance from each teacher token to its nearest student token, the $\varepsilon\to 0$ limit of the semi-relaxed formulation (Appendix~\ref{app:semirelaxed}), and scores with uniform weights at inference. \textbf{OT-uniform} is OTW with the student weights fixed to $1/K_s$, i.e., without the weight head.

\paragraph{Against score distillation.}
OTW is on par with the strongest document-dependent baseline, Listwise KL: on ColQwen3.5 it trails on v2 (60.0 vs.\ 61.2) and matches it on v1 and v3, and on Tomoro-ColQwen3-8B it leads on all three benchmarks, by 5.0 points on v3. Document-free query alignment thus matches score distillation without its page-side cache. InfoNCE trails both on both teachers: a contrastive signal from paired pages alone transfers the teacher's token geometry poorly.

\paragraph{Against fixed-weight document-free objectives.}
Without learned weights, results become teacher-dependent. OT-uniform stays within 1.7 points of OTW on ColQwen3.5 but falls 4.7--9.2 points behind on Tomoro-ColQwen3-8B. Coverage, with hard assignments, is more robust than OT-uniform yet trails OTW on every benchmark for both teachers. Adding a weight head to a trained Coverage student in a second stage closes this gap on ColQwen3.5 (Appendix~\ref{app:twostage}); OTW obtains the same result in a single stage. Uniform weights thus fit poorly when the two token sets differ, and the learned weights keep distillation stable across teachers. Appendices~\ref{app:mechanism} and~\ref{app:bound} give further analysis.

\begin{table}[t]
\centering
\footnotesize
\setlength{\tabcolsep}{5pt}
\begin{tabular}{@{}lcccc@{}}
\toprule
Pool factor $f$ & Vec./page & Teacher & Student & Ret. \\
\midrule
1 (uncompressed) & 1{,}764 & 61.6 & 59.1 & 95.8 \\
3                &    587 & 61.6 & 59.0 & 95.7 \\
9                &    195 & 61.1 & 58.2 & 95.3 \\
\bottomrule
\end{tabular}
\caption{Index compression on ViDoRe v3 (mean NDCG@5 over its 8 datasets; retention in \%). The ColVec1.1-4b index is pooled with hierarchical token pooling at pool factor $f$; the teacher and its 149M student are scored against the same pooled index.}
\label{tab:compression-main}
\end{table}

\subsection{Deployment with index compression}\label{sec:compression}

Multi-vector indices are often compressed before deployment by merging each page's tokens into fewer vectors \citep{ma2025storageefficientvisualdocumentretrieval,Kankanampati_2026}. We pool the ViDoRe v3 index of ColVec1.1-4b with hierarchical token pooling \citep{clavié2024reducingfootprintmultivectorretrieval}\footnote{Implementation from the ColPali repository, \url{https://github.com/illuin-tech/colpali}, at pool factors 3 and 9.} and score teacher and student queries against the same pooled index. The student degrades about as much as the teacher: even with 9$\times$ fewer vectors, retention drops only from 95.8\% to 95.3\% (Table~\ref{tab:compression-main}; per dataset in Appendix~\ref{app:compression}). The two savings therefore compose: a compact query encoder works with a compressed index at nearly unchanged retention.

\section{Conclusion}\label{sec:conclusion}

We introduced ColNanoVDR, a document-free, query-side distillation framework for multi-vector VDR built on OTW, which aligns query token sets by entropic optimal transport with learned token weights and bounds the MaxSim score difference on every page. Across five teachers, the 149M text-only students retain about 95\% of their teachers' NDCG@5 and encode queries up to 26$\times$ faster on a single CPU thread, while OTW matches score distillation with 12.6$\times$ less cached teacher data.

\section*{Limitations}
All students come from a single encoder family (Ettin), and every configuration is a single run with a fixed seed; we report no variance. On ColQwen3.5, the margins between document-free objectives in Table~\ref{tab:objectives} are at most two points, and those on v1 lie within the range a single run cannot resolve. The objective comparison covers two teachers, ColQwen3.5 and Tomoro-ColQwen3-8B; for the remaining three teachers in Table~\ref{tab:main} we train OTW only, so whether score distillation and the fixed-weight alternatives order the same way under them is not tested. Although OTW reads no query--page pairing, all of our training queries come from a paired set; training on unpaired query logs is not exercised. Only the query encoder is compressed: the document index and its storage remain the teacher's, and search cost falls only through the shorter query, so the method lowers query-encoding latency but not index size. The theoretical guarantee is a sufficient condition on the score discrepancy and is loose by a factor of about eight on the students we measure (Appendix~\ref{app:bound}); it does not distinguish between weightings of the student measure, and score distillation, which does not optimize it, retrieves comparably. Finally, the students are English-centric encoders evaluated on ViDoRe, whose queries are predominantly English; although the training data include machine-translated queries and ViDoRe v3 includes a French subset (Table~\ref{tab:compression}), we do not break results down by language.

\bibliography{refs}

\appendix

\section{Proofs for Section~\ref{sec:method}}\label{app:proofs}

This appendix states and proves the formal claims of
Section~\ref{sec:method} in the order in which the main text uses them.
Appendix~\ref{app:notation} fixes the notation.
Appendices~\ref{app:proof-thm1} and~\ref{app:corollary} prove
Theorem~\ref{thm:w1} and its consequence for rankings
(Section~\ref{sec:theory}). Appendix~\ref{app:cost} relates the cosine
bound to the tighter chordal one, and Appendix~\ref{app:feasibility} accounts
for the plan that the solver actually returns.
Appendix~\ref{app:semirelaxed} derives the nearest-neighbor weights of
Section~\ref{sec:marginals}, and Appendix~\ref{app:inference} shows that
token scaling realizes the weighted score at inference
(Section~\ref{sec:inference}).

\subsection{Setting and notation}\label{app:notation}

All token embeddings lie on the unit sphere of the teacher's
$m$-dimensional embedding space,
$\mathbb{S}^{m-1}\subset\mathbb{R}^m$. A page is a non-empty finite set
$D\subset\mathbb{S}^{m-1}$, and its best-match function
$h_D(x)=\max_{d\in D}\langle x,d\rangle$ is defined for every
$x\in\mathbb{R}^m$. For $x\in\mathbb{S}^{m-1}$, Cauchy--Schwarz gives
$h_D(x)\in[-1,1]$.

A weighted token set is a discrete probability measure
$\mu=\sum_{i=1}^{K}\omega_i\delta_{q_i}$ with $q_i\in\mathbb{S}^{m-1}$ and
$\omega$ in the probability simplex $\Delta_K$. Its score on a page is
$\bar S(\mu,D)=\sum_i \omega_i\,h_D(q_i)$ (Equation~\ref{eq:score-integral}).
The student measure is $\mu_S=\sum_i a_i\delta_{s_i}$ with
$a\in\Delta_{K_s}$, and the teacher measure is
$\mu_T=\sum_j b_j\delta_{t_j}$ with $b\in\Delta_{K_t}$. In our setting
$b_j=1/K_t$, but every result below holds for arbitrary $b$.
Figure~\ref{fig:sphere} illustrates the setting.

The set of transport plans is
\[
U(a,b)=\bigl\{P\in\mathbb{R}_{\ge 0}^{K_s\times K_t}:\
P\mathbf{1}=a,\ P^{\top}\mathbf{1}=b\bigr\}.
\]
It contains the product plan $ab^{\top}$ and is a compact polytope, and
every $P\in U(a,b)$ has total mass $\sum_{ij}P_{ij}=1$. We use two costs
between tokens: the chordal distance $\rho(x,y)=\lVert x-y\rVert$ and the
cosine cost $c(x,y)=1-\langle x,y\rangle$, which satisfy
$c=\tfrac12\rho^2$ on the sphere. We write $C_{ij}=c(s_i,t_j)$ and define
\begin{align*}
\W(\mu_S,\mu_T)&=\min_{P\in U(a,b)}\textstyle\sum_{ij}P_{ij}\,\rho(s_i,t_j),\\
\mathrm{OT}_c(\mu_S,\mu_T)&=\min_{P\in U(a,b)}\langle P,C\rangle .
\end{align*}
Both minima are attained, since the objectives are linear and $U(a,b)$ is
compact. The entropic plan $P^{\varepsilon}(a)$ is the unique solution of
Equation~\ref{eq:entropic-ot}. It lies in $U(a,b)$, and when $a$ and $b$
are strictly positive, as in our setting, all of its entries are strictly
positive for $\varepsilon>0$.

\begin{figure*}[t]
\centering
\includegraphics[width=0.62\textwidth]{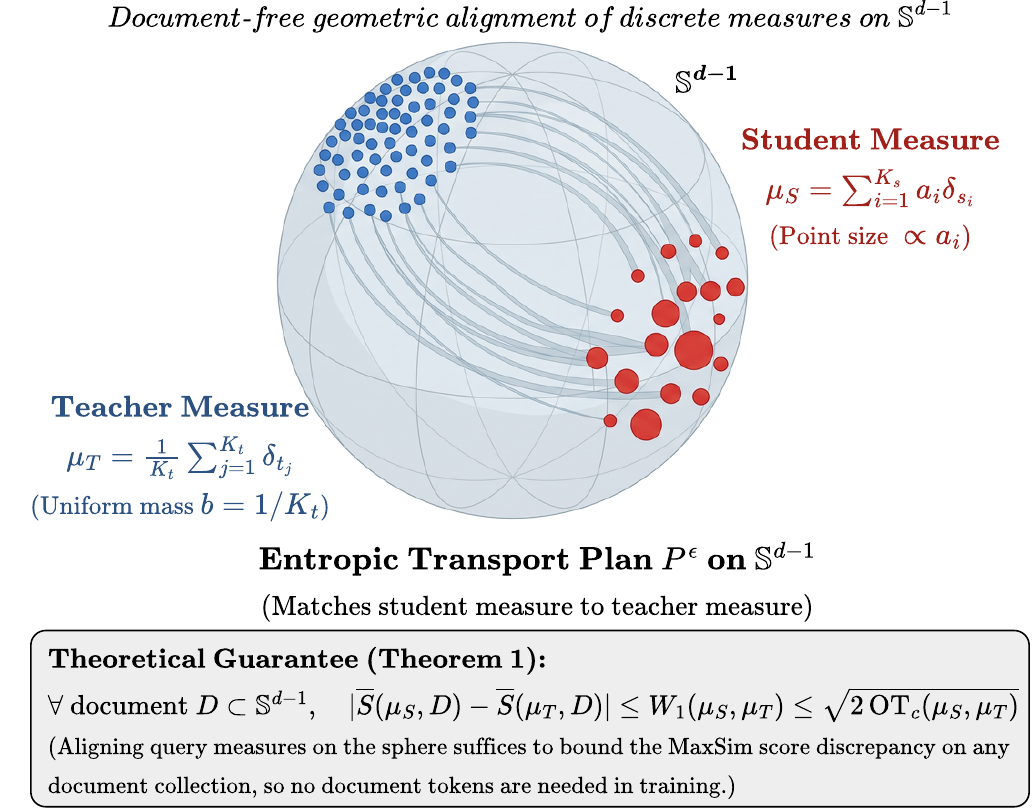}
\caption{The setting of Appendix~\ref{app:proofs}: the teacher's query
tokens form a uniform measure $\mu_T$ on the sphere, the student's a
weighted measure $\mu_S$, and a transport plan aligns them. By
Proposition~\ref{prop:anyplan}, the cost of any such plan bounds the difference
between the two weighted MaxSim scores on every page.}
\label{fig:sphere}
\end{figure*}

\subsection{Proof of Theorem~\ref{thm:w1}}\label{app:proof-thm1}

Two lemmas give a bound under the chordal distance
$\rho(x,y)=\lVert x-y\rVert$ for every transport plan
(Proposition~\ref{prop:w1}): a token's best match on a page is 1-Lipschitz
in the token (Lemma~\ref{lem:lipschitz}), and a plan rewrites the score
difference as a weighted sum over aligned pairs
(Lemma~\ref{lem:coupling}). A Cauchy--Schwarz step turns the chordal bound
into the cosine bound (Proposition~\ref{prop:anyplan}), from which
Theorem~\ref{thm:w1} follows.

\begin{lemma}\label{lem:lipschitz}
For every page $D$ and all $x,y\in\mathbb{R}^m$,
$|h_D(x)-h_D(y)|\le\lVert x-y\rVert$.
\end{lemma}

\begin{proof}
For each $d\in D$, Cauchy--Schwarz and $\lVert d\rVert=1$ give
$\langle x,d\rangle\le\langle y,d\rangle+\lVert x-y\rVert$. Let
$d^\ast\in\arg\max_{d\in D}\langle x,d\rangle$, which exists because $D$
is finite. Then
$h_D(x)=\langle x,d^\ast\rangle\le\langle y,d^\ast\rangle+\lVert x-y\rVert
\le h_D(y)+\lVert x-y\rVert$. Exchanging the roles of $x$ and $y$ gives
the reverse inequality.
\end{proof}

\begin{lemma}\label{lem:coupling}
For every $P\in U(a,b)$ and every function $f:\mathbb{R}^m\to\mathbb{R}$,
\[
\begin{aligned}
&\textstyle\sum_{i,j}P_{ij}\bigl(f(s_i)-f(t_j)\bigr)\\
&\quad=\textstyle\sum_i a_i f(s_i)-\sum_j b_j f(t_j).
\end{aligned}
\]
Taking $f=h_D$ rewrites the score difference through any plan,
\begin{equation}\label{eq:coupling}
\begin{aligned}
&\bar S(\mu_S,D)-\bar S(\mu_T,D)\\
&\qquad=\textstyle\sum_{i,j}P_{ij}\bigl(h_D(s_i)-h_D(t_j)\bigr).
\end{aligned}
\end{equation}
\end{lemma}

\begin{proof}
Summing $P_{ij}f(s_i)$ over $j$ first and using $\sum_j P_{ij}=a_i$ gives
$\sum_i a_i f(s_i)$. Summing $P_{ij}f(t_j)$ over $i$ first and using
$\sum_i P_{ij}=b_j$ gives $\sum_j b_j f(t_j)$.
\end{proof}

\begin{proposition}\label{prop:w1}
For every non-empty finite page $D$ on the unit sphere,
$\bigl|\bar S(\mu_S,D)-\bar S(\mu_T,D)\bigr|\le\W(\mu_S,\mu_T)$, where
$\W(\mu_S,\mu_T)=\min_{P\in U(a,b)}\sum_{ij}P_{ij}\,\rho(s_i,t_j)$ is the
1-Wasserstein distance under the chordal distance.
\end{proposition}

\begin{proof}
Fix a page $D$ and any $P\in U(a,b)$. By Lemma~\ref{lem:coupling}, the
nonnegativity of $P$, and Lemma~\ref{lem:lipschitz},
\begin{equation}\label{eq:anyplan-linear}
\begin{aligned}
&\bigl|\bar S(\mu_S,D)-\bar S(\mu_T,D)\bigr|\\
&\quad=\Bigl|\textstyle\sum_{i,j}P_{ij}\bigl(h_D(s_i)-h_D(t_j)\bigr)\Bigr|\\
&\quad\le\textstyle\sum_{i,j}P_{ij}\bigl|h_D(s_i)-h_D(t_j)\bigr|\\
&\quad\le\textstyle\sum_{i,j}P_{ij}\,\rho(s_i,t_j).
\end{aligned}
\end{equation}
The left-hand side does not depend on $P$, so we may take the minimum of
the right-hand side over $U(a,b)$, which is $\W(\mu_S,\mu_T)$.
\end{proof}

\begin{proposition}\label{prop:anyplan}
For every $P\in U(a,b)$ and every page $D$,
\[
\begin{aligned}
&\bigl|\bar S(\mu_S,D)-\bar S(\mu_T,D)\bigr|\\
&\quad\le\textstyle\sum_{i,j}P_{ij}\,\rho(s_i,t_j)
\le\sqrt{2\,\langle P,C\rangle}.
\end{aligned}
\]
\end{proposition}

\begin{proof}
The first inequality is Equation~\ref{eq:anyplan-linear}. For the second,
the entries of $P$ are nonnegative and sum to one, so Cauchy--Schwarz gives
\[
\begin{aligned}
&\textstyle\sum_{ij}P_{ij}\rho(s_i,t_j)\\
&\quad\le\textstyle\bigl(\sum_{ij}P_{ij}\bigr)^{1/2}
\bigl(\sum_{ij}P_{ij}\rho(s_i,t_j)^2\bigr)^{1/2},
\end{aligned}
\]
and $\rho^2=2c$ on the sphere turns the right-hand side into
$\sqrt{2\langle P,C\rangle}$.
\end{proof}

\begin{proof}[Proof of Theorem~\ref{thm:w1}]
Let $P^\ast\in U(a,b)$ attain $\mathrm{OT}_c(\mu_S,\mu_T)$, which exists
because $U(a,b)$ is compact and $\langle P,C\rangle$ is continuous.
Proposition~\ref{prop:anyplan} with $P=P^\ast$ gives the first
inequality. The entropic plan $P^{\varepsilon}(a)$ also lies in $U(a,b)$,
so $\mathrm{OT}_c(\mu_S,\mu_T)\le\langle P^{\varepsilon}(a),C\rangle
=\mathcal{L}(\theta)$, which gives the second.
\end{proof}

\begin{remark}
Two consequences are worth stating. First, a student with a low loss
scores every page nearly as its teacher does, and keeps the teacher's
order for any pair of pages whose teacher margin exceeds twice the bound
(Corollary~\ref{cor:minimax}). Second, the condition is sufficient rather
than necessary: we use it as a design principle rather than as a
predictor of retrieval quality (Appendix~\ref{app:bound}).
\end{remark}

\begin{remark}
The proofs use nothing about $D$ beyond its tokens having unit norm, so
both bounds hold uniformly, with the supremum over all pages, including
pages absent from training. Since $h_D$ is 1-Lipschitz by
Lemma~\ref{lem:lipschitz}, Proposition~\ref{prop:w1} is also an instance
of Kantorovich--Rubinstein duality \citep{villani2009optimal}. We give the
primal argument because Equation~\ref{eq:anyplan-linear} holds for every
plan, not only the optimal one, and Appendices~\ref{app:cost}
and~\ref{app:feasibility} rely on this.
\end{remark}

\subsection{Consequence for rankings}\label{app:corollary}

\begin{corollary}\label{cor:minimax}
For any pages $D$ and $D'$,
\begin{multline*}
\bigl|[\bar S(\mu_S,D)-\bar S(\mu_S,D')] \\
-[\bar S(\mu_T,D)-\bar S(\mu_T,D')]\bigr|
\;\le\;2\,\W(\mu_S,\mu_T).
\end{multline*}
In particular, if the teacher prefers $D$ to $D'$ by a margin larger than
$2\,\W(\mu_S,\mu_T)$, and a fortiori larger than
$2\sqrt{2\,\mathrm{OT}_c(\mu_S,\mu_T)}$, the student ranks them in the
same order.
\end{corollary}

\begin{proof}
By the triangle inequality, the left-hand side is at most
$|\bar S(\mu_S,D)-\bar S(\mu_T,D)|+|\bar S(\mu_S,D')-\bar S(\mu_T,D')|$,
and Proposition~\ref{prop:w1} bounds each term by $\W(\mu_S,\mu_T)$. For
the second claim, the student's margin is then at least the teacher's
margin minus $2\,\W(\mu_S,\mu_T)$, which is positive; the \emph{a
fortiori} case follows from Corollary~\ref{cor:chain}.
\end{proof}

\subsection{Chordal and cosine bounds}\label{app:cost}

Theorem~\ref{thm:w1} and the training loss use the cosine cost $c$, while
Proposition~\ref{prop:w1} uses the chordal distance $\rho$. The entropy
term of Equation~\ref{eq:entropic-ot} only smooths the plan: the loss
keeps the linear cost $\langle P^{\varepsilon},C\rangle$, because that is
what Proposition~\ref{prop:anyplan} bounds, whereas the entropic value
$\langle P^{\varepsilon},C\rangle-\varepsilon H(P^{\varepsilon})$ can be
negative and then bounds nothing. Applying
Proposition~\ref{prop:anyplan} to the optimal plans and to the entropic
plan orders the resulting bounds, which Appendix~\ref{app:bound}
measures.

\begin{corollary}\label{cor:chain}
For every $\varepsilon>0$,
\begin{align*}
\sup_D\bigl|\bar S(\mu_S,D)-\bar S(\mu_T,D)\bigr|
&\le\W(\mu_S,\mu_T)\\
&\le\sqrt{2\,\mathrm{OT}_c(\mu_S,\mu_T)}\\
&\le\sqrt{2\,\langle P^{\varepsilon}(a),C\rangle}.
\end{align*}
\end{corollary}

\begin{proof}
The first inequality is Proposition~\ref{prop:w1}. For the second, let
$P^\ast\in U(a,b)$ attain $\mathrm{OT}_c$; by the definition of $\W$ and
Proposition~\ref{prop:anyplan},
$\W\le\sum_{ij}P^\ast_{ij}\rho(s_i,t_j)\le\sqrt{2\,\mathrm{OT}_c}$. The
third holds because $P^{\varepsilon}(a)\in U(a,b)$, so
$\mathrm{OT}_c\le\langle P^{\varepsilon}(a),C\rangle$.
\end{proof}

\subsection{The plan returned by the solver}\label{app:feasibility}

Corollary~\ref{cor:chain} assumes a plan in $U(a,b)$. The solver of
Appendix~\ref{app:solver} stops after a column update, so the plan $\hat P$
it returns matches the teacher weights exactly,
$\hat P^{\top}\mathbf{1}=b$, while its row sums $\hat a=\hat P\mathbf{1}$
match $a$ only up to a truncation residual. The bound degrades gracefully
with this residual.

\begin{proposition}\label{prop:feasibility}
Let $\hat P\ge 0$ satisfy $\hat P^{\top}\mathbf{1}=b$, and let
$\hat a=\hat P\mathbf{1}$. For every page $D$,
\[
\bigl|\bar S(\mu_S,D)-\bar S(\mu_T,D)\bigr|
\le\sqrt{2\,\langle\hat P,C\rangle}+\lVert a-\hat a\rVert_1 .
\]
\end{proposition}

\begin{proof}
The entries of $\hat a$ are nonnegative and sum to $\sum_j b_j=1$, so
$\hat a\in\Delta_{K_s}$ and $\hat P\in U(\hat a,b)$. Let
$\hat\mu_S=\sum_i\hat a_i\delta_{s_i}$. Proposition~\ref{prop:anyplan},
applied to $\hat\mu_S$ and $\hat P$, gives
$|\bar S(\hat\mu_S,D)-\bar S(\mu_T,D)|\le\sqrt{2\langle\hat P,C\rangle}$.
Moreover,
$|\bar S(\mu_S,D)-\bar S(\hat\mu_S,D)|
=|\sum_i(a_i-\hat a_i)h_D(s_i)|\le\lVert a-\hat a\rVert_1$, since
$|h_D(s_i)|\le 1$. The triangle inequality combines the two.
\end{proof}

The residual $\lVert a-\hat a\rVert_1$ vanishes as the Sinkhorn iterations
converge.

\subsection{Free student weights and the coverage loss}\label{app:semirelaxed}

Section~\ref{sec:marginals} motivates the weight head by asking what the
student weights would be if they were optimized together with the plan.
Keeping only the teacher-side constraint in this way gives a semi-relaxed
transport problem, and the following proposition solves it. The same
constraint explains why the learned weights do not collapse onto a single
student token: every teacher token must still receive its full mass, so
that token would have to align with all teacher tokens, including distant
ones, at a high cost.

\begin{proposition}\label{prop:semirelaxed}
Fix the student and teacher tokens, the teacher weights $b$, and
$\varepsilon>0$. Then
\begin{multline}\label{eq:semirelaxed}
\min_{a\in\Delta_{K_s}}\ \min_{P\in U(a,b)}\ \langle P,C\rangle-\varepsilon H(P)\\
=\min_{P\ge 0,\ P^{\top}\mathbf{1}=b}\ \langle P,C\rangle-\varepsilon H(P),
\end{multline}
and the right-hand problem is solved by $P^{\varepsilon}_{ij}=b_j\,\sigma^{\varepsilon}_{ij}$,
where
$\sigma^{\varepsilon}_{ij}=e^{-C_{ij}/\varepsilon}\big/\sum_k e^{-C_{kj}/\varepsilon}$,
with induced student weights $a^{\varepsilon}_i=\sum_j b_j\,\sigma^{\varepsilon}_{ij}$.
If every teacher token has a unique nearest student token, then as
$\varepsilon\to 0$, $\sigma^{\varepsilon}_{ij}\to\mathbf{1}[\,i=\arg\min_k C_{kj}\,]$;
the induced weights converge to the share of teacher mass whose nearest
student token is $i$; and both the optimal value and the transport cost
$\langle P^{\varepsilon},C\rangle$ converge to
$\sum_j b_j\min_i C_{ij}$.
\end{proposition}

\begin{proof}
Every $P\ge 0$ with $P^{\top}\mathbf{1}=b$ has nonnegative row sums that
add up to one, so $P\in U(P\mathbf{1},b)$; conversely, every element of
some $U(a,b)$ satisfies these constraints. The two feasible sets therefore
coincide, which gives Equation~\ref{eq:semirelaxed}. The right-hand
objective separates over teacher tokens: column $j$ solves
\[
\min_{p\ge 0,\ \sum_i p_i=b_j}\ \textstyle\sum_i p_iC_{ij}+\varepsilon\sum_i p_i\log p_i ,
\]
a strictly convex problem. Stationarity of its Lagrangian gives
$C_{ij}+\varepsilon(1+\log p_i)=\lambda_j$, so $p_i\propto e^{-C_{ij}/\varepsilon}$,
and normalization gives $p_i=b_j\sigma^{\varepsilon}_{ij}$. Summing over
$j$ gives $a^{\varepsilon}$. Substituting back, the optimal value of column
$j$ is
$-\varepsilon\,b_j\log\sum_i e^{-C_{ij}/\varepsilon}+\varepsilon\,b_j\log b_j$.
As $\varepsilon\to 0$ with a unique minimizer, the softmax concentrates on
$\arg\min_i C_{ij}$, the first term tends to $b_j\min_i C_{ij}$, and the
second vanishes; the transport cost $\sum_i b_j\sigma^{\varepsilon}_{ij}C_{ij}$
has the same limit.
\end{proof}

With uniform teacher weights $b_j=1/K_t$, the limiting weights are the
nearest-neighbor shares
\begin{equation}\label{eq:nn-weights}
a^{\mathrm{NN}}_i\;=\;\tfrac{1}{K_t}\,\bigl\lvert\bigl\{\,j :\ i=\arg\max_k\langle s_k,t_j\rangle\bigr\}\bigr\rvert ,
\end{equation}
so a student token that covers several teacher tokens receives
proportionally more mass, and one that covers none receives none. The
limiting loss is
\begin{equation}\label{eq:coverage}
\mathcal{L}_{\mathrm{cov}}
=\frac{1}{K_t}\sum_{j=1}^{K_t}\Bigl(1-\max_i\langle s_i,t_j\rangle\Bigr),
\end{equation}
the coverage objective of Section~\ref{sec:mechanism}. Coverage is
therefore the $\varepsilon\to 0$ limit of transport with free student
weights, trained against the implicit weights $a^{\mathrm{NN}}$ but
deployed with uniform ones. Both $a^{\mathrm{NN}}$ and its soft version
$a^{\varepsilon}$ depend on the teacher's tokens and cannot be computed at
inference, which is why OTW predicts the weights from the query. If a
teacher token has several nearest student tokens, the limit of
$\sigma^{\varepsilon}$ splits its mass evenly among them, and
$a^{\mathrm{NN}}$ is defined accordingly.

\subsection{Weighted scoring at inference}\label{app:inference}

\begin{proposition}\label{prop:inference}
Let $a_i>0$ for every $i$, and let $\tilde s_i=a_is_i$. For every page $D$,
\[
\textstyle\sum_i\max_{d\in D}\langle\tilde s_i,d\rangle=\bar S(\mu_S,D).
\]
Consequently, the mean over the student's tokens,
$\bar S(\mu_S,D)/K_s$, ranks pages identically. Likewise, standard MaxSim
on the teacher's tokens equals $K_t\,\bar S(\mu_T,D)$ and ranks pages as
$\bar S(\mu_T,D)$ does.
\end{proposition}

\begin{proof}
Since $a_i>0$,
$\max_{d\in D}\langle a_is_i,d\rangle=a_i\max_{d\in D}\langle s_i,d\rangle
=a_i\,h_D(s_i)$; summing over $i$ gives $\bar S(\mu_S,D)$. For the teacher,
$\sum_j h_D(t_j)=K_t\sum_j\frac{1}{K_t}h_D(t_j)$. Multiplying every score
of a query by the same positive constant preserves the ranking.
\end{proof}

The softmax head of Equation~\ref{eq:weights} produces strictly positive
weights, so the proposition applies. Renormalizing $\tilde s_i$ before
scoring maps it back to $s_i$ and yields uniform-weight MaxSim, which is
why the weights must be kept in the norms of the query vectors
(Section~\ref{sec:inference}).

\section{Implementation Details}\label{app:impl}

This appendix gives the details needed to reproduce training and
evaluation, in the order in which the main text uses them: the cached
teacher query tokens that the loss consumes (Section~\ref{sec:marginals}),
the student and its optimization (Sections~\ref{sec:inference}
and~\ref{sec:setup}), the Sinkhorn solver (Section~\ref{sec:solver}), the
training data, and the evaluation protocol (Section~\ref{sec:setup}).

\begin{figure*}[t]
\centering
\begin{minipage}{0.94\textwidth}
\hrule\smallskip
\captionof{algorithm}{OTW training (one query; batches average the loss over queries)}
\label{alg:otw}
\smallskip\hrule\smallskip
\begin{algorithmic}[1]
\Require query text $Q$; cached teacher query tokens $t_1,\dots,t_{K_t}$
  (L2-normalized, from the frozen teacher); student encoder with parameters
  $\theta$: backbone $E_\theta$, projection $W_{\mathrm{proj}}$, weight head
  $w_\theta$; regularization $\varepsilon = 0.05$; iterations $N = 50$
\State $z_1,\dots,z_{K_s} \gets E_\theta(Q)$
  \Comment{token states; special tokens masked out}
\State $s_i \gets W_{\mathrm{proj}} z_i \,/\, \lVert W_{\mathrm{proj}} z_i \rVert_2$
  \Comment{student tokens on the sphere, $i = 1,\dots,K_s$}
\State $\log a \gets \operatorname{log\,softmax}\bigl(w_\theta(z_1),\dots,w_\theta(z_{K_s})\bigr)$
  \Comment{learned student weights (Eq.~\ref{eq:weights})}
\State $\log b_j \gets -\log K_t$ for all $j$
  \Comment{uniform teacher weights}
\State $C_{ij} \gets 1 - \langle s_i, t_j \rangle$
  \Comment{cosine cost, float32}
\State $f \gets \mathbf{0}$;\quad $g \gets \mathbf{0}$
\For{$n = 1,\dots,N$} \Comment{without gradient tracking}
  \State $f_i \gets \varepsilon \log a_i - \varepsilon \log\textstyle\sum_j
    \exp\bigl((g_j - C_{ij})/\varepsilon\bigr)$
    \Comment{enforces $P\mathbf{1} = a$}
  \State $g_j \gets \varepsilon \log b_j - \varepsilon \log\textstyle\sum_i
    \exp\bigl((f_i - C_{ij})/\varepsilon\bigr)$
    \Comment{enforces $P^{\top}\mathbf{1} = b$}
\EndFor
\State recompute one $(f, g)$ update pair \emph{inside the autograd graph}
  \Comment{one-step gradient}
\State $P_{ij} \gets \exp\bigl((f_i + g_j - C_{ij})/\varepsilon\bigr)$
  \Comment{transport plan (Eq.~\ref{eq:plan-recovery})}
\State $\mathcal{L} \gets \textstyle\sum_{ij} P_{ij}\, C_{ij}$
  \Comment{gradients flow through $C$ and $\log a$}
\State update $\theta$ by AdamW on $\mathcal{L}$
\Statex \emph{Inference:} scale each student token by its weight $a_i$ and
  run MaxSim against the teacher's document index, unchanged
  (Proposition~\ref{prop:inference}).
\end{algorithmic}
\smallskip\hrule
\end{minipage}
\end{figure*}

\subsection{Teacher caches and query tokenization}\label{app:tokens}

\paragraph{Teacher caches.}
Teacher embeddings are precomputed once in float16 at each teacher's own
width (320 or 640 dimensions). For ColQwen3.5 we cache the 711{,}603
training pairs (queries and page images), the 777{,}649 translated query
variants, and all evaluation query and corpus sets. We cache the same
parts for Tomoro-ColQwen3-8B, the second teacher of the objective
comparison; for the other three teachers of Table~\ref{tab:main} we cache
the same queries and evaluation sets and no training page. Those teachers are Tomoro-ColQwen3-8B (8.8B
parameters, 320 dimensions), Vultron-4.5B (4.5B, 320), and
ColVec1.1-4b and -8b (4.5B and 8.4B, 640). For every teacher, tokens are
retained by nonzero norm, which keeps the ColBERT-style query augmentation
tokens and drops the padding rows a model may emit inside the attention
mask. Visual inputs use each teacher's own processor: ColQwen3.5 caps a
page at 768 visual tokens (749 per page on average over the evaluation
corpora); Tomoro-ColQwen3-8B caps at 1{,}280 and Vultron-4.5B at 1{,}792, caps
the evaluation pages do not reach, so both encode a page into 1{,}227
tokens on average; and the two ColVec models resize pages to 1{,}792
tokens (1{,}709 on average). For the ViDoRe v1 datasets, corpora are
deduplicated by image hash in first-appearance order, and cached
embeddings follow the same order as evaluation.

\paragraph{Token counts.}
On the training queries, ColQwen3.5 produces 28.6 tokens per query on
average against 17.1 for the student. Over the evaluation suites, the
teacher averages 27.9/31.5/37.2 tokens on v1/v2/v3 against 16.5/27.2/33.9
for the student. Excluding the teacher's ten augmentation tokens, the
student produces 0.92/1.26/1.25 times as many tokens as the teacher on
average.

\paragraph{Extrapolated cache for ColVec1.1.} ColVec1.1 encodes a page into 1{,}792 visual tokens (1{,}709 on average over the evaluation corpora) at 640 dimensions, about 2.1\,MiB per page in float16. Assuming one distinct page per pair, one million training pairs would therefore require about 2.0\,TiB of page cache, whereas its measured query cache (67\,GiB for 1{,}489{,}252 queries) corresponds to about 45\,GiB per million queries, a ratio of roughly 45. These figures are estimates; we did not cache training pages for this teacher.

\subsection{Student architecture and optimization}\label{app:training}

Each student is an Ettin encoder \citep{weller2026seqvsseqopen} with the
two heads of Section~\ref{sec:inference}: a bias-free linear projection to
the teacher's width followed by L2 normalization, and a weight head that is
a single linear layer to one logit per token. The tokenizer's [CLS],
[SEP], and padding positions are masked; every other position inside the
attention mask is a valid token.

\paragraph{Optimization.}
Every run trains for 10 epochs with an effective batch size of 1{,}024
(128 per step, gradient accumulation 4, two H200 GPUs), using AdamW with
weight decay $0.01$, a peak learning rate of $3{\times}10^{-4}$ under a
one-cycle cosine schedule with $3\%$ warmup, and mixed precision. The
transport objective uses $\varepsilon=0.05$ with 50 log-domain Sinkhorn
iterations and a full-precision cost matrix (Appendix~\ref{app:solver});
sensitivity to $\varepsilon$ is reported in Appendix~\ref{app:eps}.
Listwise KL normalizes MaxSim scores by query length so that a fixed
temperature stays calibrated; it uses temperatures of 0.07 for the teacher and 0.05 for the
student, InfoNCE a temperature of 0.05, and both take as negatives the
other pages of the same 128-query micro-batch. All results are single runs
with seed 42.

A run takes under 30 wall-clock hours on two H200 GPUs. Per epoch, OTW
trains at a speed comparable to Listwise KL and InfoNCE. The Sinkhorn
iterations of Appendix~\ref{app:solver} act on small per-query matrices,
roughly $17\times29$ on the training queries with ColQwen3.5 and at most
$17\times36$ with the other teachers.

\begin{table*}[t]
\centering
\small
\begin{tabular}{lrrr}
\toprule
Part & Items & Token vectors & Size \\
\midrule
queries, base       & 711{,}603 & 20{,}388{,}709 & 12.2\,GiB \\
queries, translated & 777{,}649 & 25{,}398{,}098 & 15.1\,GiB \\
page images         & 711{,}603 & 529{,}150{,}162 & 315.4\,GiB \\
\bottomrule
\end{tabular}
\caption{ColQwen3.5 teacher cache written to disk for training (float16,
320-dimensional tokens). The translated variants add queries only.}
\label{tab:cache}
\end{table*}

\subsection{Sinkhorn solver}\label{app:solver}

\paragraph{Iterations.}
The solution of Equation~\ref{eq:entropic-ot} has the form
$P^{\varepsilon}_{ij}=u_i\,e^{-C_{ij}/\varepsilon}\,v_j$, and the Sinkhorn
algorithm finds it by alternately rescaling the rows and the columns of
$e^{-C/\varepsilon}$ until they sum to $a$ and $b$: a softmax over
similarities, normalized along both axes. For numerical stability at small
$\varepsilon$ we run it in the log domain, with $\varepsilon=0.05$ and 50
iterations, and compute the cost matrix $C_{ij}=1-\langle s_i,t_j\rangle$
in float32 under mixed-precision training. The solver maintains potentials
$f\in\mathbb{R}^{K_s}$ and $g\in\mathbb{R}^{K_t}$, initialized at zero and
alternately updated as
\begin{equation}\label{eq:sinkhorn-updates}
\begin{aligned}
f_i &\leftarrow \varepsilon \log a_i
- \varepsilon \log \textstyle\sum_j \exp\bigl(\tfrac{g_j - C_{ij}}{\varepsilon}\bigr),\\
g_j &\leftarrow \varepsilon \log b_j
- \varepsilon \log \textstyle\sum_i \exp\bigl(\tfrac{f_i - C_{ij}}{\varepsilon}\bigr),
\end{aligned}
\end{equation}
with numerically stable log-sum-exp reductions. The plan is recovered in
closed form,
\begin{equation}\label{eq:plan-recovery}
P^{\varepsilon}_{ij} \;=\;
\exp\bigl(\tfrac{f_i + g_j - C_{ij}}{\varepsilon}\bigr).
\end{equation}
Each update enforces its own constraint exactly: after the $f$ update,
$P\mathbf{1}=a$; after the $g$ update, $P^{\top}\mathbf{1}=b$. Masked
positions carry zero cost and negative infinity in the log-marginals.

\paragraph{Gradients.}
The updates of Equation~\ref{eq:sinkhorn-updates} run without gradient
tracking. One final update pair is then recomputed inside the autograd
graph from the converged potentials, and the loss of
Equation~\ref{eq:final-objective} is evaluated from
Equation~\ref{eq:plan-recovery}. Gradients reach the student through two
inputs: the cost matrix $C$, a function of the token embeddings $s_i$, and
the log-marginal $\log a$, a function of the weight-head logits. This
one-step scheme is a truncation of implicit Sinkhorn differentiation
\citep{eisenberger2022unifiedframeworkimplicitsinkhorn}. By an envelope
argument, it is exact for the gradient of the entropic optimal value,
whereas the exact gradient of the transport cost
$\langle P^{\varepsilon},C\rangle$ would require solving an additional
linear system at the fixed point
\citep{luise2018differentialpropertiessinkhornapproximation}. We use the
one-step form, which avoids backpropagating through all 50 iterations and
which we found sufficient in development runs. Because each pair ends with the $g$ update, the
returned plan satisfies the column constraint exactly and the row
constraint up to a truncation residual, whose effect on the bound
Proposition~\ref{prop:feasibility} controls. Algorithm~\ref{alg:otw}
summarizes the procedure.

\subsection{Training data}\label{app:data}

The 711{,}603 base pairs comprise the VisRAG synthetic set (234K, 32.9\%)
and the VisRAG in-domain set (94K, 13.2\%)
\citep{yu2025visragvisionbasedretrievalaugmentedgeneration}, the VDR
multilingual set covering five languages (275K, 38.6\%)
\citep{cimolai2025vdr}, and the ColPali training set (109K, 15.3\%)
\citep{faysse2025colpaliefficientdocumentretrieval}. The 777{,}649 translated variants are machine translations of base
queries that reuse the corresponding page images.

\begin{table}[t]
\centering
\small
\begin{tabular}{lccc}
\toprule
Student & v1 & v2 & v3 \\
\midrule
\multicolumn{4}{l}{\textit{ColQwen3.5-4.5B teacher}} \\
\quad Ettin-32M  & 89.8 (98.0) & 55.2 (86.7) & 51.0 (87.0) \\
\quad Ettin-68M  & 90.4 (98.7) & 59.2 (92.9) & 54.0 (92.0) \\
\quad Ettin-150M & 90.7 (99.0) & 60.0 (94.2) & 55.1 (93.8) \\
\quad Ettin-400M & 91.2 (99.5) & 62.1 (97.4) & 56.7 (96.5) \\
\multicolumn{4}{l}{\textit{Tomoro-ColQwen3-8B teacher}} \\
\quad Ettin-32M  & 88.8 (98.0) & 56.7 (87.2) & 51.3 (86.9) \\
\quad Ettin-68M  & 89.8 (99.1) & 59.6 (91.7) & 54.1 (91.6) \\
\quad Ettin-150M & 90.0 (99.3) & 60.6 (93.3) & 54.9 (93.0) \\
\quad Ettin-400M & 90.3 (99.6) & 62.5 (96.1) & 56.4 (95.5) \\
\bottomrule
\end{tabular}
\caption{Retention across the Ettin family for both ablation teachers
(NDCG@5; \% retention of that teacher).}
\label{tab:scaling}
\end{table}

\begin{table*}[t]
\centering
\small
\setlength{\tabcolsep}{5pt}
\begin{tabular}{lrccccccccc}
\toprule
 & & \multicolumn{3}{c}{Uncompressed ($f=1$)} & \multicolumn{3}{c}{$f=3$} & \multicolumn{3}{c}{$f=9$} \\
\cmidrule(lr){3-5}\cmidrule(lr){6-8}\cmidrule(lr){9-11}
Dataset & Pages & Teacher & Student & Ret. & Teacher & Student & Ret. & Teacher & Student & Ret. \\
\midrule
Finance (en) & 2{,}942 & 66.9 & 63.3 & 94.6 & 66.5 & 62.8 & 94.5 & 66.3 & 62.8 & 94.8 \\
Finance (fr) & 2{,}384 & 48.9 & 47.4 & 96.8 & 48.8 & 47.2 & 96.8 & 48.2 & 46.1 & 95.8 \\
Computer sci. & 1{,}360 & 78.1 & 74.2 & 95.1 & 77.9 & 74.2 & 95.2 & 77.6 & 73.5 & 94.8 \\
Human res. & 1{,}110 & 64.8 & 61.6 & 95.1 & 64.7 & 61.6 & 95.2 & 64.5 & 61.0 & 94.6 \\
Energy & 2{,}225 & 66.4 & 65.2 & 98.2 & 66.6 & 64.8 & 97.4 & 65.9 & 63.8 & 96.9 \\
Industrial & 5{,}244 & 54.8 & 51.4 & 93.8 & 55.2 & 51.4 & 93.1 & 54.8 & 51.1 & 93.2 \\
Pharmaceutical & 2{,}313 & 65.3 & 64.0 & 97.9 & 65.2 & 63.9 & 97.9 & 64.4 & 63.0 & 97.8 \\
Physics & 1{,}674 & 47.9 & 45.5 & 95.0 & 48.0 & 45.9 & 95.8 & 47.3 & 44.6 & 94.3 \\
\midrule
\textit{Average} & 19{,}252 & 61.6 & 59.1 & 95.8 & 61.6 & 59.0 & 95.7 & 61.1 & 58.2 & 95.3 \\
\bottomrule
\end{tabular}
\caption{Index compression on ViDoRe v3 (NDCG@5; retention in \%). The
ColVec1.1-4b index is pooled with hierarchical token pooling at pool factor
$f$, and the teacher and its 149M student are scored against the same
pooled index. Average vectors per page: 1{,}764 ($f=1$), 587 ($f=3$), 195
($f=9$).}
\label{tab:compression}
\end{table*}

\subsection{Evaluation}\label{app:eval}

Reference systems and teachers in Table~\ref{tab:main} are the following
public checkpoints on HuggingFace:
\par\smallskip\noindent{\footnotesize\raggedright\ttfamily
TomoroAI/tomoro-colqwen3-embed-8b\\
athrael-soju/colqwen3.5-4.5B-v3\\
vultr/VultronRetrieverCore-Qwen3.5-4.5B\\
webAI-Official/webAI-ColVec1.1-4b\\
webAI-Official/webAI-ColVec1.1-8b\\
TomoroAI/tomoro-colqwen3-embed-4b\\
nomic-ai/colnomic-embed-multimodal-7b\\
MrLight/dse-qwen2-2b-mrl-v1\\
vidore/colpali-v1.3\\
ModernVBERT/colmodernvbert\\
nanovdr/NanoVDR-Q-DistilBERT-\\\hspace*{1em}Qwen3VL2B-2048-ML\par}\smallskip
NDCG@5 is computed with \texttt{pytrec\_eval} and averaged per benchmark
over 10 (v1), 4 (v2), and 8 (v3) datasets. Student scores use mean-MaxSim
over query tokens with each token scaled by its weight, which realizes
$\bar S(\mu_S,D)/K_s$ and ranks pages as $\bar S(\mu_S,D)$ does
(Proposition~\ref{prop:inference}). Teacher ceilings are computed from the
same cached embeddings and index.

\section{Efficiency and Capacity}\label{app:efficiency}

This appendix reports the measurements behind Sections~\ref{sec:results} and~\ref{sec:compression}.

\subsection{Query-encoding cost}\label{app:latency}

Table~\ref{tab:serving} reports the measurements and
Figure~\ref{fig:teaser}a plots its CPU column as throughput. All systems are measured on the same
node: CPU numbers on an Intel Xeon Platinum 8562Y+ with a single thread,
float32, and batch size 1; GPU numbers on one H200 in bf16. Each number is
the median over 20 queries after 3 warmup runs, and none includes MaxSim
scoring. For our students the document index is identical to the teacher's by
construction. ColQwen3.5's GPU number
reflects the torch fallback for its hybrid linear-attention layers.
ColNanoVDR students of the same size share the encoder and differ only in
projection width, so Table~\ref{tab:serving} reports one row per size.

\subsection{Supervision cost}\label{app:cache}

Table~\ref{tab:cache} measures the caches that the objectives read during
training. Document tokens dominate: the 711{,}603 page images produce 26
times as many cached token vectors as their queries, because a page is
encoded into up to 768 visual tokens plus prompt text while a query is a few
dozen tokens. Document-dependent objectives read the page part at every step
(Listwise KL also the query part); the document-free objectives read only
the query part, which is 8.0\% of the
cache. Including the translated variants, the page cache is 11.6 times the
size of the query cache, so score distillation reads 12.6 times as much
cached teacher data as the document-free objectives (342.7 against
27.3\,GiB), the ratio quoted in Section~\ref{sec:results}.

Encoding the training set with ColQwen3.5 ran as eight shard jobs of 2.4
to 4.0 hours each on one H200 (23 GPU-hours in total), with page tokens
making up 96\% of the tokens written. These jobs encoded queries and pages
together; the query-only cost is bounded by a separate job that encoded
the 777{,}649 translated queries and all evaluation sets in 2 hours 2
minutes on one H200. The pages involved are those of the distillation set,
not the deployment index, which the teacher encodes for retrieval
regardless of how the student is trained. A new teacher therefore pays the
page-side cost again before score distillation can start, and only the
query-side cost for the document-free recipe. For the three teachers
cached without pages, encoding the 1{,}489{,}252 training and translated
queries took 1.0 to 1.7 hours on one H200 each, model loading included,
and the resulting query caches occupy 27\,GiB at 320 dimensions and
67\,GiB at 640.

\subsection{Student capacity}\label{app:capacity}

Table~\ref{tab:scaling} lists the retention numbers behind
Figure~\ref{fig:efficiency_scale}.

\subsection{Index compression}\label{app:compression}

Table~\ref{tab:compression} lists the numbers behind
Section~\ref{sec:compression}. The student is the 149M ColVec1.1-4b
student of Table~\ref{tab:main}, and the index is that teacher's cached
ViDoRe v3 corpus embeddings: 19{,}252 pages at 640 dimensions, 1{,}764
tokens per page on average (33.95M vectors in total).

\paragraph{Pooling.}
We use the hierarchical token pooling of \citet{clavié2024reducingfootprintmultivectorretrieval},
through \texttt{HierarchicalTokenPooler} in the \texttt{colpali-engine}
library with its default settings. For each page separately, it computes the
cosine distances $1-\langle d_k,d_l\rangle$ between the page's tokens, builds
a Ward agglomerative clustering on them, and cuts it into at most
$\lfloor n/f\rfloor$ clusters, where $n$ is the page's token count and $f$
the pool factor. Each cluster is replaced by the mean of its tokens,
renormalized to unit length. Pages are pooled independently, so no token
is shared across pages. The pooled index keeps 11.31M vectors at $f=3$ (587
per page) and 3.76M at $f=9$ (195 per page), realized compression ratios of
3.00 and 9.02. Pooling is applied to the index once, offline, in float32
on CPU; it touches neither encoder.

\paragraph{Scoring.}
Teacher queries are the teacher's cached query embeddings; student queries
are encoded by the student and scaled by their weights, as in
Appendix~\ref{app:eval}. Both are scored with exhaustive MaxSim against the
same pooled index, with no approximate search, and evaluated with the
protocol of Appendix~\ref{app:eval}. Retention is the student's NDCG@5
divided by the teacher's on the same pooled index. The uncompressed rows
($f=1$) reproduce the ColVec1.1-4b entries of Table~\ref{tab:main}.

\section{Analyses of the Training Objectives}\label{app:mechanism}

This appendix asks where the advantage of OTW over the fixed-weight
alternatives of Section~\ref{sec:mechanism} comes from.
Appendix~\ref{app:weighted-kl} tests whether the weight head also helps
score distillation, and Appendix~\ref{app:twostage} a two-stage alternative in which the
weights are trained after the embeddings.
Appendices~\ref{app:protocol}--\ref{app:eps} then fix trained students and
vary only the weights used at inference, to measure how much the deployed
weights matter, how they compare with teacher-derived oracle weights, and
how sensitive both are to $\varepsilon$. Every student here is distilled
from ColQwen3.5.

\subsection{Weighted score distillation}\label{app:weighted-kl}

The weight head is part of OTW but not of score distillation, so a
learned-weight variant of Listwise KL separates the objective from the
architecture. Training Listwise KL with the same head, its weights scaling the
student's tokens inside the MaxSim score as at inference, gives
90.9/60.9/54.9 on v1/v2/v3 against 90.9/61.2/55.0 for uniform-weight Listwise KL:
within 0.3 points on every benchmark, and no higher on any. The
head therefore does not transfer the gain it produces under OTW
(Table~\ref{tab:objectives}) to score distillation, which is why
Table~\ref{tab:objectives} compares the two families at their own settings
rather than crediting the head to both.

\begin{table}[t]
\centering
\footnotesize\setlength{\tabcolsep}{3.5pt}
\begin{tabular}{lccc}
\toprule
Head target & v1 & v2 & v3 \\
\midrule
hard assignment & 90.7 (99.0) & 60.0 (94.1) & 55.1 (93.9) \\
soft assignment & 90.7 (99.0) & 60.1 (94.3) & 55.1 (93.9) \\
OTW, joint      & 90.7 (99.0) & 60.0 (94.2) & 55.1 (93.8) \\
\bottomrule
\end{tabular}
\caption{Two-stage weight heads on frozen coverage geometry, 149M student
(NDCG@5; \% retention). The OTW row is the jointly trained reference.}
\label{tab:twostage}
\end{table}

\begin{table*}[t]
\centering
\small
\begin{tabular}{lcccccc}
\toprule
& \multicolumn{3}{c}{149M student} & \multicolumn{3}{c}{395M student} \\
Weights & v1 & v2 & v3 & v1 & v2 & v3 \\
\midrule
Learned (default)   & 90.7 & 60.0 & 55.1 & 91.2 & 62.1 & 56.7 \\
Dead tokens pruned  & 90.7 & 60.0 & 54.9 & 91.1 & 62.2 & 56.4 \\
Uniform             & 89.1 & 56.5 & 51.6 & 89.8 & 60.8 & 54.8 \\
Inverted (dead only)& 73.8 & 27.0 & 26.2 & 80.0 & 34.3 & 31.2 \\
\bottomrule
\end{tabular}
\caption{Inference-time weight interventions (NDCG@5): embeddings fixed,
only the per-token weights replaced; dead tokens as defined in
Appendix~\ref{app:protocol}.}
\label{tab:intervention}
\end{table*}

\subsection{Two-stage weight heads}\label{app:twostage}

Table~\ref{tab:twostage} freezes the encoder of the 149M student trained with coverage
and trains only a weight head for two epochs, supervised per query by
either the hard assignment (cross-entropy to the nearest-neighbor shares of
Equation~\ref{eq:nn-weights}) or the soft assignment $a^{\varepsilon}$ of
Proposition~\ref{prop:semirelaxed} at $\varepsilon=0.05$. Either head
reproduces the jointly trained OTW student to within 0.1 points on
ColQwen3.5. This two-stage recipe is the only configuration in this paper
that is not trained in a single stage, which is why
Table~\ref{tab:objectives}, whose variants are all single-stage, does not list
it. It is simple to add on top of an existing coverage student, and we
retain it as an engineering alternative. OTW remains our default: it
trains the encoder and the weights end to end in one stage, keeps a single
training pipeline, and needs fewer epochs (10, against 10 + 2 for the
two-stage recipe). Appendix~\ref{app:intervention} measures the effect of
the deployed weights directly.

\subsection{Protocol for inference-time analyses}\label{app:protocol}

Section~\ref{sec:mechanism} compares objectives end to end: each row of
Table~\ref{tab:objectives} is a separately trained student. The analyses
of Appendices~\ref{app:intervention} and~\ref{app:eps} use a different
control. They take one trained student, freeze its token embeddings, and
replace only the weight vector used at inference, so that a difference in
NDCG@5 is attributable to the weights alone. Retrieval numbers are NDCG@5
over the complete benchmark, with the same protocol and cached teacher
embeddings as Table~\ref{tab:main}. A student token is \emph{dead} if it
is the nearest student token of no teacher token of its query, that is, if
$a^{\mathrm{NN}}_i=0$ in Equation~\ref{eq:nn-weights}. Every weight vector
other than \emph{learned} and \emph{uniform} is computed from the teacher's
query tokens and is therefore an oracle that cannot be deployed.

\subsection{Weight interventions}\label{app:intervention}

Table~\ref{tab:intervention} replaces the learned weights of the 149M and
395M OTW students at inference. Pruning the dead tokens changes almost
nothing, so the head has effectively already removed them. Reverting to
uniform weights costs 1.6/3.5/3.5 points on the 149M student and
1.4/1.3/1.9 on the 395M student, more than the gap between OTW and
OT-uniform under separate training (Table~\ref{tab:objectives}): the
embeddings of a student trained with learned weights are adapted to them.
Placing all mass on the dead tokens collapses retrieval.

\begin{table}[t]
\centering
\footnotesize
\begin{tabular}{lcc}
\toprule
Weights & 149M (coverage) & 395M \\
\midrule
hard assignment & 55.2 & 56.6 \\
soft, $\varepsilon=0.005$ & 55.2 & 56.6 \\
soft, $\varepsilon=0.02$ & 55.2 & 56.7 \\
soft, $\varepsilon=0.05$ & 55.3 & 56.7 \\
soft, $\varepsilon=0.1$ & 55.3 & 56.6 \\
soft, $\varepsilon=0.2$ & 55.0 & 56.3 \\
soft, $\varepsilon=0.5$ & 54.5 & 55.6 \\
soft, $\varepsilon=1.0$ & 54.3 & 55.2 \\
flattened, $T=2$ & 54.8 & 56.1 \\
flattened, $T=4$ & 54.5 & 55.5 \\
flattened, $T=8$ & 54.2 & 55.2 \\
flattened, $T\to\infty$ & 54.0 & 54.9 \\
uniform & 53.9 & 54.8 \\
learned & -- & 56.7 \\
\bottomrule
\end{tabular}
\caption{Inference-weight variants on fixed geometries, NDCG@5 on ViDoRe v3.
``149M (coverage)'' is the 149M student trained with coverage; ``395M'' is
the 395M OTW student. All variants except \emph{learned} and \emph{uniform}
use the teacher's query tokens at inference.}
\label{tab:smooth-full}
\end{table}

\begin{table*}[t]
\centering
\small
\begin{tabular}{llrrrrrr}
\toprule
Objective & Reads & \multicolumn{1}{c}{$\sup_D|\Delta \bar S|$}
 & \multicolumn{1}{c}{centered} & \multicolumn{1}{c}{$\rho_s$}
 & \multicolumn{1}{c}{$\W$} & \multicolumn{1}{c}{$\sqrt{2\mathrm{OT}_c}$}
 & \multicolumn{1}{c}{$\sqrt{2\langle P^{\varepsilon}\!,C\rangle}$} \\
\midrule
\quad InfoNCE & D & 0.335 & 0.251 & 0.780 & 1.256 & 1.262 & 1.287 \\
\quad Listwise KL & Q+D & 0.157 & 0.106 & 0.936 & 0.913 & 0.925 & 0.953 \\
\quad Coverage & Q & 0.084 & 0.079 & 0.949 & 0.613 & 0.656 & 0.694 \\
\quad OT-uniform & Q & 0.083 & 0.082 & 0.945 & 0.656 & 0.674 & 0.723 \\
\quad OTW & Q & 0.069 & 0.059 & 0.956 & 0.568 & 0.596 & 0.621 \\
\bottomrule
\end{tabular}
\caption{Measured discrepancy against the bounds of
Corollary~\ref{cor:chain} for the ColQwen3.5 students of
Table~\ref{tab:objectives}, medians over the analysis sample.
``centered'' removes the per-query mean offset; $\rho_s$ is the Spearman
correlation between the student's and the teacher's document scores;
``Reads'' as in Table~\ref{tab:objectives}.}
\label{tab:bound}
\end{table*}

\subsection{Oracle weights and the role of $\varepsilon$}\label{app:eps}

Table~\ref{tab:smooth-full} sweeps teacher-derived oracle weights on two
fixed geometries, the 149M student trained with coverage and the 395M OTW
student. The variants are the hard assignment of Equation~\ref{eq:nn-weights};
the soft assignment $a^{\varepsilon}$ of Proposition~\ref{prop:semirelaxed},
in which each teacher token spreads its mass over student tokens in
proportion to $e^{-C_{ij}/\varepsilon}$; and temperature-flattened versions
of the hard assignment. The hard, soft, and learned weights lie within 0.1
points of one another for $\varepsilon\le 0.1$, and the hard and soft
assignments beat uniform weights by 1 to 2 points; the soft assignment
degrades toward uniform as $\varepsilon$ grows past $0.2$. The learned
weights therefore recover what a reasonable teacher-derived weighting
would give, without the teacher, and the choice of $\varepsilon$ is not
delicate below the training value.

Training-time sensitivity is of the same size. Retraining the 149M student
with $\varepsilon=0.02$ gives 90.8/59.1/55.1 on v1/v2/v3, and with
$\varepsilon=0.1$ gives 90.3/59.5/54.2, against 90.7/60.0/55.1 at the
default $0.05$. The default is within a point of either neighbor on every
benchmark and best or tied on two of three.

\section{Empirical Check of the Bound}\label{app:bound}

Theorem~\ref{thm:w1} is a sufficient condition, used in
Section~\ref{sec:method} as a design principle. This appendix measures
whether the chain of Corollary~\ref{cor:chain} holds on trained students,
how loose it is, and whether the alignment cost tracks retrieval
quality.

\paragraph{Protocol.}
We use a fixed sample of 4{,}735 evaluation queries, the \emph{analysis
sample}: arxivqa (500) and docvqa (451) from v1, biomedical\_lectures (640)
from v2, and finance\_en (1{,}854) and cs (1{,}290) from v3, chosen to span
the three benchmarks. For every query of this sample and each
single-objective ColQwen3.5 student of Table~\ref{tab:objectives} (the
learned-weight KL student of Appendix~\ref{app:weighted-kl} is not
included), we score every page of the corresponding corpus with the
student's and the teacher's measures, using the weights the student
deploys (learned for OTW, uniform for the others), and take
$\sup_D|\bar S(\mu_S,D)-\bar S(\mu_T,D)|$. We compare it with the three
bounds of Corollary~\ref{cor:chain}: $\W$, computed exactly by linear
programming under the chordal distance; $\sqrt{2\,\mathrm{OT}_c}$,
approximated by entropic transport at $\varepsilon=0.005$ with 300
iterations; and $\sqrt{2\,\langle P^{\varepsilon},C\rangle}$ at the
training setting. Atoms carrying less than $10^{-6}$ of the student's mass
are dropped before the linear program, which changes $\W$ by at most twice
the dropped mass; every query solves. Table~\ref{tab:bound} reports
medians over queries.

\paragraph{The bound holds and is not vacuous.}
The chordal bound $\W$ (Proposition~\ref{prop:w1}) and the training-setting bound
$\sqrt{2\langle P^{\varepsilon},C\rangle}$ hold on every query of every
student. The intermediate $\sqrt{2\,\mathrm{OT}_c}$, which we only
approximate, falls below $\W$ on a minority of queries (at most 12\%, for
OT-uniform, and by at most 0.07), a residual of the approximation rather
than a failure of the chain. It is not vacuous: $|\bar S|\le 1$
makes $2$ the trivial bound, and the document-free students measure $\W$
between $0.57$ and $0.66$. It is not tight either: for those students $\W$
exceeds the worst-case discrepancy it certifies by $7.3$ to $8.3$ times.
Most of that slack is already present in the chordal bound: passing to
the cosine cost of Theorem~\ref{thm:w1} inflates it by a further $3$ to
$7$ percent, and the entropic term at the training setting
by $4$ to $7$ percent. Removing the per-query offset, which shifts every
page equally and cannot change a ranking, changes the discrepancy of the
document-free students by less than $20$ percent.

\paragraph{Orderings.}
Among the document-free objectives, $\W$ orders the students as their
average NDCG@5 does, OTW below coverage below OT-uniform
(Table~\ref{tab:objectives}; the single benchmark-level exception is v2),
and OTW has the highest Spearman correlation with the teacher's document
scores. The document-dependent objectives, in which nothing drives the two
measures together, leave $\W$ at $0.91$ and $1.26$, yet Listwise KL reaches a
Spearman correlation of $0.94$: a student can score accurately while its
measure remains far from the teacher's. This is consistent with the bound
being sufficient rather than necessary (Section~\ref{sec:theory}).

\end{document}